\documentclass[submission,copyright,creativecommons]{eptcs}
\providecommand{\event}{ICLP 2023}

\usepackage{iftex}

\ifpdf
	\usepackage{underscore}         % Only needed if you use pdflatex.
	\usepackage[T1]{fontenc}        % Recommended with pdflatex
\else
	\usepackage{breakurl}           % Not needed if you use pdflatex only.
\fi

\usepackage{times}
\usepackage{soul}
\usepackage{url}
\usepackage[utf8]{inputenc}
\usepackage[small]{caption}
\usepackage{graphicx}
\usepackage{algorithm}
\usepackage{algorithmic}
\usepackage{latexsym}
\usepackage{amssymb}
\usepackage{amsmath}
\usepackage{amsthm}
\usepackage{esvect}
\usepackage{xcolor}
\usepackage{xspace}
\usepackage{booktabs}
\usepackage{mathtools}
\usepackage{stmaryrd}
\usepackage{multirow}
\usepackage{tikz}
\usetikzlibrary{shapes,arrows,positioning,patterns,fit,decorations.markings}
\usepackage{listings}
\usepackage{adjustbox}
\usepackage{enumitem}
\usepackage{xparse}
\usepackage{bm}
\usepackage{pgffor}

\foreach \x in {A,...,Z}{%
		\expandafter\xdef\csname cal\x\endcsname{\noexpand\ensuremath{\noexpand\mathcal{\x}}}
	}

\newcommand\T{\rule{0pt}{2.6ex}}       % Top strut
\newcommand\B{\rule[-1.2ex]{0pt}{0pt}} % Bottom strut

\newcommand{\ol}[1]{\ensuremath{\overline{#1}}}
\renewcommand{\neg}[1]{\ensuremath{\ol{#1}}\xspace}
\newcommand{\lambdasol}{\ensuremath{\lambda}\xspace}
\newcommand{\rowprefix}[1]{\ensuremath{{#1}{:}\ } }
\newcommand{\lit}{\ensuremath{L}\xspace}
\newcommand{\litid}[1]{\ensuremath{\lit_{#1}}\xspace}
\newcommand{\at}{\ensuremath{A}\xspace}
\newcommand{\litset}{\ensuremath{S}\xspace}
\newcommand{\litsetx}{\ensuremath{X}\xspace}
\newcommand{\nottt}{{\thicksim}}
\DeclareMathOperator{\when}{\leftarrow}
\newcommand{\head}[1][r]{\ensuremath{H(#1)}\xspace}
\newcommand{\body}[1][r]{\ensuremath{B(#1)}\xspace}
\newcommand{\pAtoms}{\ensuremath{\calA}\xspace}
\newcommand{\pp}{\ensuremath{\calP}\xspace}
\newcommand{\ppp}{\ensuremath{\pp^\star}\xspace}
\newcommand{\bodyp}[2][r]{\ensuremath{B^{#2}(#1)}\xspace}
\newcommand{\bodypos}[1][r]{\bodyp[#1]{+}}
\newcommand{\bodyneg}[1][r]{\bodyp[#1]{-}}
\newcommand{\LitA}{\ensuremath{Lit_\pAtoms}\xspace}
\newcommand{\ppEx}{\ensuremath{\pp_{ex}}\xspace}
\newcommand{\conf}{\ensuremath{\gamma}\xspace}
\newcommand{\rul}{\ensuremath{r}\xspace}
\newcommand{\rulid}[1]{\ensuremath{\rul_{#1}}\xspace}
\newcommand{\rulrow}[1]{\ensuremath{{\rulid{#1}}{:}\ } }
\newcommand{\inst}{\ensuremath{\calF}\xspace}
\newcommand{\instex}{\ensuremath{\inst_{ex}}\xspace}
\newcommand{\gop}{\ensuremath{G}\xspace}
\newcommand{\cgb}{\ensuremath{\bm\Gamma}\xspace}
\newcommand{\cg}{\ensuremath{\Gamma}\xspace}
\newcommand{\cgr}[1]{\ensuremath{\cg(#1)}\xspace}
\newcommand{\cgid}[1]{\ensuremath{\cg_{#1}\xspace}}
\newcommand{\cgsize}[1]{\ensuremath{w_{#1}\xspace}}
\newcommand{\cgordop}[1]{\ensuremath{>_{#1}}\xspace}
\newcommand{\cgord}{\ensuremath{\cgordop{\cgcov}}\xspace}
\newcommand{\cgordprime}{\ensuremath{\cgordop{\cgcovprime}}\xspace}
\newcommand{\cgcov}{\ensuremath{\cgb}\xspace}
\newcommand{\cq}{\ensuremath{Q}\xspace}
\newcommand{\cqb}{\ensuremath{\mathbf{\cq}}\xspace}
\newcommand{\cqid}[1]{\ensuremath{\cq_{#1}}\xspace}
\newcommand{\confcol}[1]{\ensuremath{{#1}_{\mathit{cf}}}\xspace}
\newcommand{\ppconfcol}{\confcol{\pp}}
\newcommand{\ccg}[1]{\ensuremath{\mathit{CCG}(#1)}\xspace}
\newcommand{\ccgpp}{\ccg{\pp}\xspace}
\newcommand{\ccgppp}{\ccg{\ppp}\xspace}
\newcommand{\ccgppex}{\ccg{\ppEx}\xspace}
\newcommand{\cgcovex}{\ensuremath{\cgb_{ex}}\xspace}
\newcommand{\gcgcov}{\ensuremath{\gop(\cgcov)}\xspace}
\newcommand{\gcgcovex}{\ensuremath{\gop({\cgcovex})}\xspace}
\newcommand{\gcgcovprime}{\ensuremath{\gop({\cgcovprime})}\xspace}
\newcommand{\gcgcovprimeprime}{\ensuremath{\gop({\cgcovprimeprime})}\xspace}
\newcommand{\lblop}{\ensuremath{\mathit{lb}}\xspace}
\newcommand{\lambdaord}[1]{\ensuremath{\succ_{#1}}\xspace}
\newcommand{\cgcovprime}{\ensuremath{\cgb^\prime}\xspace}
\newcommand{\cgcovprimeprime}{\ensuremath{\cgb^ {\prime\prime}}\xspace}
\newcommand{\clq}[1]{\ensuremath{\mathit{CLQ}(#1)}\xspace}
\newcommand{\clqcg}{\ensuremath{\clq{\cgcov}}\xspace}
\newcommand{\cqweight}{\ensuremath{\omega\xspace}}
\newcommand{\vrts}{\ensuremath{V}\xspace}
\newcommand{\edges}{\ensuremath{E}\xspace}
\newcommand{\edgeslbl}[1]{\ensuremath{\edges^{#1}}\xspace}
\newcommand{\gcgclq}{\ensuremath{\gop(\cgcov,\lblop)}\xspace}

\newcommand{\exA}{\ensuremath{a}\xspace}
\newcommand{\exB}{\ensuremath{b}\xspace}
\newcommand{\exC}{\ensuremath{c}\xspace}
\newcommand{\exD}{\ensuremath{d}\xspace}
\newcommand{\exE}{\ensuremath{e}\xspace}
\newcommand{\exF}{\ensuremath{f}\xspace}
\newcommand{\exG}{\ensuremath{g}\xspace}
\newcommand{\exH}{\ensuremath{h}\xspace}

\newcommand{\exJ}{\ensuremath{j}\xspace}
\newcommand{\exK}{\ensuremath{k}\xspace}
\newcommand{\exL}{\ensuremath{l}\xspace}
\newcommand{\exM}{\ensuremath{m}\xspace}
\newcommand{\exN}{\ensuremath{n}\xspace}
\newcommand{\exO}{\ensuremath{o}\xspace}
\newcommand{\exP}{\ensuremath{p}\xspace}
\newcommand{\exQ}{\ensuremath{q}\xspace}

\newcommand{\exS}{\ensuremath{s}\xspace}
\newcommand{\exT}{\ensuremath{t}\xspace}
\newcommand{\exU}{\ensuremath{u}\xspace}

\newcommand{\exW}{\ensuremath{w}\xspace}
\newcommand{\exX}{\ensuremath{x}\xspace}
\newcommand{\exY}{\ensuremath{y}\xspace}
\newcommand{\exZ}{\ensuremath{z}\xspace}

\newcommand{\exSympM}{\ensuremath{\textit{sympM}}\xspace}
\newcommand{\exSympN}{\ensuremath{\textit{sympN}}\xspace}
\newcommand{\exSympO}{\ensuremath{\textit{sympO}}\xspace}
\newcommand{\exCondA}{\ensuremath{\textit{disA}}\xspace}
\newcommand{\exCondB}{\ensuremath{\textit{disB}}\xspace}
\newcommand{\exTreatX}{\ensuremath{\textit{treatX}}\xspace}
\newcommand{\exTreatY}{\ensuremath{\textit{treatY}}\xspace}

\newcommand{\ie}{i.\,e.\xspace }
\newcommand{\cf}{cf.\xspace }

\newcommand{\sT}{s.\,t.\xspace }
\newcommand{\viz}{viz.\xspace }
\newcommand{\wrt}{w.r.t.\xspace }
\newcommand{\contd}{contd.\xspace }

\newtheorem{example}{Example}

\newtheorem{definition}{Definition}
\newtheorem{proposition}{Proposition}

\title{Sorting Strategies for Interactive Conflict Resolution in ASP}
\author{Andre Thevapalan
	\institute{Technische Universität Dortmund, Dortmund, Germany}
	\email{andre.thevapalan@tu-dortmund.de}
	\and
	Gabriele Kern-Isberner
	\institute{Technische Universität Dortmund, Dortmund, Germany}
	\email{gabriele.kern-isberner@cs.tu-dortmund.de}
}

\newcommand{\titlerunning}{Sorting Strategies for Interactive Conflict Resolution in ASP}
\newcommand{\authorrunning}{Andre Thevapalan \& Gabriele Kern-Isberner}

\hypersetup{
	bookmarksnumbered,
	pdftitle    = {\titlerunning},
	pdfauthor   = {\authorrunning},
	pdfsubject  = {EPTCS},               % Consider adding a more appropriate subject or description
}
\begin{document}
\maketitle

\begin{abstract}
	Answer set programs in practice are often subject to change.
	This can lead to inconsistencies in the modified program due to conflicts between rules which are the results of the derivation of strongly complementary literals.
	To facilitate the maintenance of consistency in answer set programs, in this paper
	we continue work on a recently presented framework that implements interactive conflict resolution by extending the bodies of conflicting rules by suitable literals, so-called \emph{\lambdasol-extensions}.
	More precisely, we present strategies to choose \lambdasol-extensions that allow for resolving several conflicts at a time in an order that aims at minimizing (cognitive) efforts.
	In particular, we present a graphical representation of connections between conflicts and their possible solutions.
	Such a representation can be utilized to efficiently guide the user through the conflict resolution process by displaying conflicts and suggesting solutions in a suitable order.
\end{abstract}

\section{Introduction}
Answer set programming provides valuable features for usage in real-world applications where complex decisions have to be made, thanks to its declarative nature and the availability of both strong negation and default negation.
Such programs, though, are often subject to change.
Adding rules to an answer set program can however potentially lead to inconsistency due the activation of conflicting rules, \ie, rules with complementary literals in their heads.%to contradictory statements inside the program.
The approach in~\cite{ThevapalanKernIsberner2022} deals with inconsistency caused by the derivation of strongly complementary literals.
Herefore, the notion of \lambdasol-extensions for conflicting rules has been introduced that enables an interactive conflict resolution process in which a knowledge expert can help restore the consistency of an updated program in a professionally adequate way.
A \lambdasol-extension is a set of (default) literals with which the body of a conflicting rule can be extended to constrain its applicability and thus resolves the conflict.
However, the paper \cite{ThevapalanKernIsberner2022} focuses on conflicts between two rules only. But conflicts can involve more than two rules, and solutions to one conflict can affect solutions to other conflicts. It is clear  that no new conflicts may arise by extending rule bodies.  So we might expect even synergetic positive effects when considering all interactions between conflicts in a program at the same time and finding a clever order to solve conflicts. This is exactly the topic here.

In this work, we extend the approach of \cite{ThevapalanKernIsberner2022}  by first defining a graph that shows the connections between the conflicts and their solutions, thereby embedding the conflicts of a program and their possible solutions within the overall context of the program's conflict resolution.
This graphical representation of possible solutions is then utilized to define suitable orders over conflicts and the respective \lambdasol-extensions so that one solution can help solving subsequent other conflicts. In particular, the strategies presented in this paper can also help shrink the (sometimes large) sets of possible  \lambdasol-extensions by choosing such extensions that are involved in more than one conflict. In this way, also the cognitive burden for knowledge experts during the conflict resolution process can be reduced.

The main contributions of this paper can be summarized as follows:
\begin{itemize}
	\item We introduce a suitable graph structure named \lambdasol-graphs for \lambdasol-extensions to display the relationships between conflicts \wrt their possible solutions.
	\item We show how \lambdasol-graphs provide the necessary, syntax-based information to find suitable orders over conflicts.
	\item Utilizing \lambdasol-graphs, we furthermore explain how one can obtain an order over the possible solutions of a conflict.
	\item Based on these results, we present an explicit sorting strategy that defines an order over conflicts and their solutions for the application in a conflict resolution framework as proposed in~\cite{ThevapalanKernIsberner2020,ThevapalanKernIsberner2022}.
\end{itemize}

This paper begins by laying out the necessary preliminaries in Section~\ref{sec:preliminaries}.
Section~\ref{sec:conflicts-and-lambdasol-extensions} provides further terminology regarding conflict resolution which is then used in Section~\ref{sec:relationship-between-lambdasol-extensions} to construct \lambdasol-graphs that provide crucial information regarding the connections between conflicts and \lambdasol-extensions.
Based on these results in Section~\ref{subsec:order-strategy}, a sorting strategy is proposed that defines an order over conflicts and their possible solutions by taking their relationships among each other into account.
Section~\ref{sec:related-work} gives a brief overview of related work.
%previous approaches that also tackle the task of restoring consistency in logic programs.
We conclude this paper by summarizing our findings in Section~\ref{sec:conclusion-and-future-work} and briefly outlining possible future work.

%===========================================================
\section{Preliminaries}%
\label{sec:preliminaries}
%===========================================================
In this paper, we look at non-disjunctive \emph{extended logic programs} (ELPs)~\cite{GelfondLifschitz1991}.
An ELP is a finite set of rules over a set $\calA$ of propositional atoms.
A literal \lit is either an atom \at (\emph{positive literal}) or a negated atom $\neg\at$ (\emph{negative literal}).
For a literal \lit, the \emph{strongly complementary} literal \neg\lit is $\neg\at$ if ${\lit = \at}$ and \at otherwise.
A \emph{default-negated literal} \lit, called \emph{default literal}, is written as $\nottt\lit$.
Given a set \litset of literals, we say a literal \emph{\lit is true in \litset} (symbolically ${\litset \vDash \lit}$) iff ${\lit \in \litset}$ and \emph{$\nottt \lit$} is true in \litset (symbolically ${\litset \vDash \nottt \lit}$) iff ${\lit \notin \litset}$.
A set of literals is \emph{inconsistent} if it contains strongly complementary literals.

We are now ready to specify the form of ELPs.

A \emph{rule} \rul is of the form
\begin{align}\label{eqn:rule}
	\litid{0} \when \litid{1}, \dots, \litid{m}, \nottt \litid{m+1}, \dots, \nottt \litid{n}.,
\end{align}
with literals $\litid{0}, \dots, \litid{n}$ and \mbox{$0 \leq m \leq n$}.
The literal $\litid{0}$ is the \emph{head} of \rul, denoted by $\head$, and $\{\litid{1}, \dots \litid{m},$ $\nottt \litid{m+1}, \dots \nottt \litid{n}\}$ is the \emph{body} of \rul, denoted by \body.
Furthermore, $\{\litid{1}, \dots, \litid{m}\}$ is denoted by $\bodypos$ and $\{\litid{m+1}, \dots, \litid{n}\}$ by $\bodyneg$.
%For a head literal \head, the \emph{head atom} is the atom on which the literal in \head is based on.
%A rule \rul with $\body = \emptyset$ is called a \emph{fact}, and \rul is called a \emph{constraint} if it has an empty head.
An \emph{extended logic program (ELP)} is a set of rules of the form (\ref{eqn:rule}).
Given a set \litset of literals, a rule $r$ \emph{is true in \litset} (symbolically ${\litset \vDash \rul}$) iff \head is true in \litset whenever \body is true in \litset.
In this case we also say that \litset \emph{satisfies} \rul.
Given an ELP \pp over \pAtoms without default negation and $\LitA=\{\pAtoms \cup \{\neg \at \mid \at \in \pAtoms\}\}$, the \emph{answer set} of \pp is either (a) the smallest set $\litset \subseteq \LitA$ such that \litset is consistent and $\litset \vDash \rul$ for every rule $\rul \in \pp$, or (b) the set \LitA of literals in case all such sets are inconsistent.
Note that, similar to Horn logic programs, each such ELP without default negation has exactly one minimal model which might, however, be inconsistent.

In general, an \emph{answer set of an ELP \pp} is determined by its reduct.
The \emph{reduct} $\pp^\litset$ of a program \pp relative to a set $\litset$ of literals is defined by
\begin{align*}
	\pp^\litset = \{\head \when \bodypos. \mid \rul \in \pp, \bodyneg \cap \litset = \emptyset \}.
\end{align*}
A set \litset of literals is an \emph{answer set} of \pp if it is the answer set of $\pp^\litset$~\cite{GelfondLifschitz1991}.

%===========================================================
\section{Conflicts and \lambdasol-extensions}%
\label{sec:conflicts-and-lambdasol-extensions}
%===========================================================

In~\cite{ThevapalanKernIsberner2020,ThevapalanKernIsberner2022}, a framework is outlined that supports knowledge experts in restoring consistency interactively in answer set programs.
To this aim, the bodies of rules involved in a conflict are extended suitably by known literals, so-called \lambdasol-extensions. In this way, conflicts that cause inconsistency can be resolved.
In particular, a method is provided that for every conflict generates all possible \lambdasol-extensions, from which the expert can then choose the most adequate ones.

We briefly recall the basic techniques from \cite{ThevapalanKernIsberner2020,ThevapalanKernIsberner2022}, in particular the terms \emph{conflicts} and \emph{\lambdasol-extensions}, and illustrate them with the following (running) example.
\begin{example}\label{ex:run}
	Let \ppEx be specified by the following rules:
	\begin{align*}
		 & \rulrow{1} \exA \when \exB,\nottt \exC.            &
		 & \rulrow{2} \neg \exA \when \exB.                   &
		 & \rulrow{3} \exX \when \exD,\exE,\exF,\nottt \exC.  &
		 & \rulrow{4} \neg \exX \when \exD,\exE.              & \\
		 & \rulrow{5} \exY \when \exG,\exH,\exF.              &
		 & \rulrow{6} \neg \exY \when \exG.                   &
		 & \rulrow{7} \exZ \when \exJ,\exK,\nottt \exL.       &
		 & \rulrow{8} \neg \exZ \when \exJ,\nottt \exL.       & \\
		 & \rulrow{9} \exW \when \exF,\exM,\exN.              &
		 & \rulrow{10} \neg \exW \when \exM.                  &
		 & \rulrow{11} \neg \exW \when \exN.                  &
		 & \rulrow{12} \exP \when \exO,\exH,\exF,\nottt \exQ. & \\
		 & \rulrow{13} \neg \exP \when \exO,\nottt \exQ.      &
		 & \rulrow{14} \exU \when \exS.                       &
		 & \rulrow{15} \neg \exU \when \exS,\neg \exT,\exH.   &
		 & \rulrow{16} \neg \exU \when \exS,\exT,\exH.
	\end{align*}
\end{example}

Note that \ppEx in Example~\ref{ex:run} is trivially consistent because it is a so-called \emph{program core}~\cite{ThevapalanKernIsberner2022}, i.e., it has no facts.
Program cores are usable with different instances by expanding the program by a corresponding set \inst of facts (instance data).
A typical example would be a medical expert system where the instance data are provided by the patients and so cannot be part of the generic program.
\begin{example}\label{ex:simple}
	Suppose the following program \pp describing the symptoms of and treatments for two diseases \exCondA and \exCondB:
	\begin{align*}
		 & \exCondA \when \exSympM, \exSympN. &
		 & \exCondB \when \exSympM, \exSympO. &
		 & \exTreatX \when \exCondA.          &
		 & \exTreatY \when \exCondB.
	\end{align*}
	Program \pp is a program core.
	Adding the instance data $\inst = \{\exSympM, \exSympO\}$ as patient data to \pp yields the unique answer set $\inst\cup\{\exCondB,\exTreatY\}$ describing that the corresponding patient has condition \exCondB and should be treated with treatment \exTreatY.
\end{example}

However, it is easy to see that there exist multiple instance data \inst for \ppEx that would yield an inconsistent program $\ppEx\cup\inst$ because of rules in \ppEx that, if their bodies are satisfied simultaneously, derive strongly complementary literals.

\begin{example}[Example~\ref{ex:run} \contd]\label{ex:run:conflict}
	Let \instex be a set of instance data for \ppEx such that $\exB \in \instex$ and $\exC \not\in \instex$.
	Then in $\ppEx\cup\instex$, both $\rulid{1}$ and $\rulid{2}$ are simultaneously satisfied.
	Hence, $\ppEx\cup\instex$ is inconsistent as both \exA and $\neg\exA$ are derivable.
\end{example}

Rule sets like $\{r_1,r_2\}$ of the previous example are called \emph{conflicting rules} or \emph{conflicts}.

\begin{definition}[Conflict and conflict group~(\cf~\cite{ThevapalanKernIsberner2020})]
	In a logic program \pp, two rules $\rul,\rul^\prime$ are \emph{conflicting} iff there exists a set $\litset$ of literals such that \litset satisfies $\body$ and $\body[\rul^\prime]$ simultaneously and the head literals $\head$ and $\head[\rul^\prime]$ are strongly complementary.
	A \emph{conflict} is a set $\conf = \{\rul,\rul^\prime\}$ of two conflicting rules $\rul,\rul^\prime$.
	For a rule \rul, the corresponding \emph{conflict group} $\cgr{\rul}$ is the set of all conflicts \conf in \pp with $\rul \in \conf$.
	The \emph{size of a conflict group} is the number of different conflicts in a conflict group.
\end{definition}

Intuitively, two rules are conflicting if their bodies can be true simultaneously and their head literals are contradictory.
The following example illustrates conflict groups and their sizes in \ppEx.

\begin{example}[Example~\ref{ex:run} \contd]\label{ex:run:groups}
	The conflict groups for \rulid{1} and \rulid{4} are the sets $\cgr{\rulid{1}} = \{\{\rulid{1},\rulid{2}\}\}$ and $\cgr{\rulid{4}} = \{\{\rulid{3},\rulid{4}\}\}$, respectively.
	Both groups have each size 1 while conflict group $\cgr{\rulid{14}} = \{\{\rulid{14},\rulid{15}\},$ $\{\rulid{14},\rulid{16}\}\}$ has size 2.
\end{example}

If both rules of a conflict in a program core are simultaneously satisfiable, a program is potentially inconsistent~\cite{ThevapalanKernIsberner2022}.
In the rest of this paper, we deal with programs that are inconsistent due to \emph{conflicts} and we require a program to be coherent, that is, each given program has at least one answer set (\cf~\cite{CostantiniIntrigilaProvetti2003}).

In order to guarantee that a program has an answer set whenever it is extended by consistent instance data, all conflicts have to be resolved.
For that, we use an approach called \emph{\lambdasol-extensions}.
In the following, we will present the main aspects of \lambdasol-extensions and properties that are relevant to this work.
For more details, we refer the reader to~\cite{ThevapalanKernIsberner2022}.

\begin{definition}[\lambdasol-extensions, \cf\cite{ThevapalanKernIsberner2022}]
	Suppose a program \pp and a rule $\rul\in\pp$ that is in conflict with a non-empty set of rules $R \subset \pp$.
	Let \litsetx be a set of (default) literals built from atoms occurring in $R$, and let $\rul^{\prime\prime}$ be a rule obtained from \rul where $\body$ is extended by \litsetx, \viz, $\rowprefix{\rul^{\prime\prime}}\head \when \body \cup \litsetx.$.
	This set \litsetx is a \emph{(conflict-resolving) \lambdasol-extension for rule \rul} if for every rule $\rul^\prime\in R$ it holds that $\rul^\prime$ and $\rul^{\prime\prime}$ are no longer conflicting.
	Such a rule $\rul^{{\prime\prime}}$ is called a \emph{\lambdasol-extended rule \wrt \litsetx}.
	A \lambdasol-extension \litsetx for $\rul$ is \emph{minimal} iff there exists no set $\litsetx^\prime\subset\litsetx$ such that $\litsetx^\prime$ is also a \lambdasol-extension for \rul.
	A conflict group \cgr{\rul} is called \emph{resolvable} if there exists a rule $\rul^\prime$ and a \lambdasol-extension \litsetx for $\rul^\prime$ such that replacing $\rul^\prime$ in every conflict of \cgr{\rul} by the corresponding \lambdasol-extended rule \wrt \litsetx leads to all conflicts in \cgr{\rul} being resolved.
	Rule $\rul^\prime$ is then called the \emph{representative of~\cgr{\rul}}.
\end{definition}

Intuitively, a \lambdasol-extension \litsetx for a representative $\rul$ of a resolvable conflict group \cgr{\rul} is a set of (default) literals such that if the body of \rul is expanded by \litsetx, any previous conflicts of \rul are resolved.
In~\cite{ThevapalanKernIsberner2022}, the authors show that a conflict $\{\rul,\rul^\prime\}$ can only be resolved iff there exists at least one atom in $\body[\rul^\prime]$ that is not in \body or vice versa.
They also demonstrate that resolving each conflict in a program core \pp using \lambdasol-extensions yields a \emph{uniformly non-contradictory} program core \ppp, meaning, the program core can be used with any set of consistent instance data that consists of literals that only appear in rule bodies of \pp.

We will illustrate the workings of \lambdasol-extensions in the following example.
\begin{example}[Example~\ref{ex:run:groups} \contd]\label{ex:run:multiple}
	Consider conflict groups \cgr{\rulid{14}} in Example~\ref{ex:run} which consists of the conflicts $\{\rulid{14},\rulid{15}\}$ and $\{\rulid{14},\rulid{16}\}$.
	For these conflicts, we get the \lambdasol-extensions  $\{\nottt\exH\}$, $\{\neg\exH\}$ and  $\{\nottt\exT,\nottt \neg\exT\}$ as possible extensions for the rule body of \rulid{14}.
	This means both conflicts of \rulid{14} can be solved by replacing \rulid{14} in \ppEx by one of the following rules:
	\begin{align*}
		 & {\rulid{14}'}{:\ } \exU \when \exS,\nottt\exH.                &
		 & \rulid{14}''{:\ } \exU \when \exS,\neg\exH.                   &
		 & \rulid{14}'''{:\ } \exU \when \exS,\nottt\neg\exT,\nottt\exT.
	\end{align*}
\end{example}

\begin{table}[tb]
	\centering
	\caption{Components of \lambdasol-graph}\label{tab:tab-p}
	{\begin{tabular}{@{\extracolsep{\fill}}llllll}
			\hline
			\textit{Group}              & \textit{Representative} & \textit{Conflicts}                                    & \textit{\lambdasol-Ext.}                           & \textit{Node}         & \textit{Cliques} \T\B        \\ \hline
			\textbf{$\cgr{\rulid{2}}$}  & \rulid{2}               & $\{\rulid{1},\rulid{2}\}$                             & $\{\exC\}$                                         & ($\cgr{\rulid{2}}$,1) & \cqid{3}                  \T \\
			\textbf{$\cgr{\rulid{4}}$}  & \rulid{4}               & $\{\rulid{3},\rulid{4}\}$                             & $\{\nottt \exF\},\{\exC\}$                         & (\cgr{\rulid{4}},1)   & \cqid{2},\cqid{3}            \\
			\textbf{$\cgr{\rulid{6}}$}  & \rulid{6}               & $\{\rulid{5},\rulid{6}\}$                             & $\{\nottt \exH\},\{\nottt \exF\}$                  & (\cgr{\rulid{6}},1)   & \cqid{1},\cqid{2}            \\
			\textbf{$\cgr{\rulid{8}}$}  & \rulid{8}               & $\{\rulid{7},\rulid{8}\}$                             & $\{\nottt  \exK\}$                                 & (\cgr{\rulid{8}},1)   & \cqid{4}                     \\
			\textbf{$\cgr{\rulid{10}}$} & \rulid{10}              & $\{\rulid{9},\rulid{10}\}$                            & $\{\nottt  \exF\}$                                 & (\cgr{\rulid{10}},1)  & \cqid{1}                     \\
			\textbf{$\cgr{\rulid{11}}$} & \rulid{11}              & $\{\rulid{9},\rulid{11}\}$                            & $\{\nottt  \exF\}$                                 & (\cgr{\rulid{11}},1)  & \cqid{1}                     \\
			\textbf{$\cgr{\rulid{13}}$} & \rulid{13}              & $\{\rulid{12},\rulid{13}\}$                           & $\{\nottt \exH\},\{\nottt \exF\}$                  & (\cgr{\rulid{13}},1)  & \cqid{1},\cqid{2}            \\
			\textbf{$\cgr{\rulid{14}}$} & \rulid{14}              & $\{\rulid{14},\rulid{15}\},\{\rulid{14},\rulid{16}\}$ & $\{\nottt \exH\},\{\nottt \neg \exT,\nottt \exT\}$ & (\cgr{\rulid{14}},2)  & \cqid{2},\cqid{5} \B         \\
			\hline
		\end{tabular}}
\end{table}

Note that no additional conflicts are introduced as the premise for \rul is solely extended yielding a more specific condition.
However, not every conflict group of a program is resolvable.
The following example shows that in order to utilize conflict groups for the resolution of conflicts, their proper selection is a crucial step.

\begin{example}[Example~\ref{ex:run:groups} \contd]\label{ex:run:extensions}
	For rules \rulid{1} and \rulid{2} it holds that $\cgr{\rulid{1}} = \cgr{\rulid{2}} = \{\{\rulid{1},\rulid{2}\}\}$.
	As $\body[\rulid{2}] \subseteq \body[\rulid{1}]$ holds, \rulid{1} can not be picked as a representative.
	However, $\body[\rulid{1}]\backslash\body[\rulid{2}] = \{\nottt \exC\}$ holds and thus the \lambdasol-extension $\{\exC\}$ resolves the conflict of both conflict groups.
	Similarly, regarding \cgr{\rulid{4}}, we get the \lambdasol-extensions $\{\neg\exF\}$, $\{\nottt\exF\}$, and $\{\exC\}$ for \rulid{4}.
	Note that conflict group $\cgr{\rulid{9}}$ has no possible representative as there does not exist any \lambdasol-extension for \rulid{9} that solves all conflicts in \cgr{\rulid{9}}.
	It is therefore necessary to pick the conflict groups \cgr{\rulid{10}} and \cgr{\rulid{11}} which are both resolvable via the representative \rulid{10} and \rulid{11} respectively.
	Table~\ref{tab:tab-p} shows the \lambdasol-extensions of the different resolvable conflict groups of \ppEx.
\end{example}

In the rest of this paper for any atom $a$, we omit the \lambdasol-extension that contains $\neg a$ whenever $\nottt a$ and $\neg a$ can both be used because $\nottt a$ is more cautious.
In Example~\ref{ex:run:extensions}, we henceforth solely state $\{c\}$ and $\{\nottt f\}$ as the \lambdasol-extensions for \rulid{4}.

Furthermore, note that conflicts like $\{\exA\when\exB., \neg\exA\when\exB.\}$ require an expansion of both rules bodies by complementary literals in order to be resolved.
We argue that suitable expansions for such conflicts can only be determined by the knowledge expert.
Thus, conflicts of this type will not be considered in this paper.

%===========================================================
\section{Relationship between \lambdasol-extensions}%
\label{sec:relationship-between-lambdasol-extensions}
%===========================================================
Naturally, a practical implementation of the interactive framework for conflict resolution, as presented in \cite{ThevapalanKernIsberner2020,ThevapalanKernIsberner2022}, has to provide a proper workflow to resolve each conflict and suggest solutions in a suitable order.
We propose that a syntax-based approach considers the different connections between possible solutions in order to condense the resolution of multiple conflicts.
For this reason, we introduce \emph{\lambdasol-graphs} and corresponding \emph{clique covers} that can be used to point out such connections.
Based on these results in Section~\ref{sec:sort-conflict-using-lambda-graphs}, we demonstrate how \lambdasol-graphs and clique covers can be used to define explicit strategies that specify in which order conflicts and solutions should be presented to the knowledge expert.

For resolving conflicts in a program thoroughly, we have to make sure that each rule that is involved in a conflict must be taken into regard. This idea is formalized by conflict group covers.

\begin{definition}[Conflict group cover]\label{def:conflictgroupcover}
	Let $\confcol{\pp}$ be the set of all rules in \pp that are part of a conflict, and let \cgcov be a set of resolvable conflict groups.
	Then, \cgcov is a \emph{conflict group cover of \pp} if each rule in \confcol{\pp} appears in at least one conflict group of~\cgcov.
	The set of all inclusion-minimal conflict group covers of a logic program \pp is denoted by \ccgpp.
\end{definition}

A conflict group cover \cgcov of \pp therefore implies via the representative of each conflict group in \cgcov which rules in \pp shall be modified.

Notice that a conflict group cover \cgcov involves a sufficient set of rules that have to be modified to resolve every conflict.

\begin{proposition}\label{prop:groupcover}
	Given a conflict group cover \cgcov, expanding the body of each rule $\rul$ in $\{\rul\mid \text{ there is } \cgr{\rul'}\in\cgcov \text{ such that }\rul \text{ is a representative of }\cgr{\rul'}\}$ by one of its respective \lambdasol-extensions yields a consistent program.
\end{proposition}

\begin{proof}
	Let \cgcov be a conflict group cover of a program \pp with conflicts and \ppconfcol the set of conflicting rules in \pp.
	By Definition~\ref{def:conflictgroupcover}, a set $\cgcov\in\ccgpp$ is a set of conflict groups such that every rule in \ppconfcol appears in at least one conflict group $\cg\in\cgcov$.
	Since every conflict group in \cgcov is resolvable by Definition~\ref{def:conflictgroupcover}, there exists at least one \lambdasol-extension for the representative of each conflict group in \cgcov.
	This in turn means that applying a respective \lambdasol-extension to the representative of each conflict group in \cgcov resolves every conflict in \pp, thus, yielding a conflict-free program $\pp^\prime$.
\end{proof}

This result implies that regarding conflicting rules, it is mandatory for the knowledge expert to initially decide which rules are allowed to be modified and which rules must stay unaffected.
This allows to determine all appropriate conflict groups and consequently all appropriate conflict groups covers.
Choosing the most suitable cover then ensures that a sufficient and moreover the most suitable set of rules can be modified using \lambdasol-extensions to obtain a uniformly non-contradictory program.

We can now define  \lambdasol-graphs \wrt a conflict group cover \cgcov by making use of weights for the nodes, and labels for the edges to store information which is crucial for the resolution process.

\begin{definition}[\lambdasol-graph]\label{def:graph}
	Given a conflict group cover $\cgcov\in\ccgpp$ of a logic program \pp, the \emph{\cgcov-induced \lambdasol-graph \gcgcov} is a tuple $(\vrts,\edges)$ of weighted nodes \vrts and labeled edges \edges where \vrts contains a weighted node $(\cg,\cgsize{\cg})$ for each conflict group \cg in \cgcov with size $\cgsize{\cg}$, and \edges contains labeled edges $(\cg,\cg^\prime,\lambdasol)$ whenever the representatives of two different conflict groups \cg and $\cg^\prime$ have a common \lambdasol-extension \lambdasol, and for any \lambdasol-extension \lambdasol that is an extension for a representative of only one conflict group \cg, there exists a self-loop $(\cg,\cg,\lambdasol)$.
\end{definition}

Since every conflict group in a program \pp is represented by a weighted node and each node has either an edge to itself or to another node that shares a common \lambdasol-extension, every conflict is considered in a \lambdasol-graph.

We now illustrate \cgcov-induced \lambdasol-graphs using the running example.

\begin{example}[Example~\ref{ex:run} \contd]\label{ex:run:induce}
	Suppose \ppEx $= \confcol{\pp}$ and the conflict group cover $\cgcovex =$ $\{\cgr{\rulid{2}},$ $\cgr{\rulid{4}},$ $\cgr{\rulid{6}},$ $\cgr{\rulid{8}},$ $\cgr{\rulid{10}},$ $\cgr{\rulid{11}}, $ $\cgr{\rulid{13}},$ $\cgr{\rulid{14}}\} \in \ccgppex$ are given.
	Table~\ref{tab:tab-p} shows the conflict groups in \cgcovex and the \lambdasol-extensions of the corresponding representatives as stated in Table~\ref{tab:tab-p}.
	The \cgcovex-induced \lambdasol-graph $\gcgcovex =$ $(\vrts,\edges)$ is obtained in the following way:
	for each conflict group $\cg$ and its size \cgsize{\cg}, we define $(\cg,\cgsize{\cg})$ as the corresponding weighted node.
	Then, \vrts consists of all such weighted nodes, \viz, $\vrts =$ $\{(\cgr{\rulid{2}},1),$ $(\cgr{\rulid{4}},1),$ $(\cgr{\rulid{6}},1),$ $(\cgr{\rulid{8}},1),$ $(\cgr{\rulid{10}},1),$ $(\cgr{\rulid{11}},1),$ $(\cgr{\rulid{13}},1),$ $(\cgr{\rulid{14}},2)\}$.
	The set \edges consists of all labeled edges $(\cgid{},\cgid{}^\prime,\lblop)$ such that \cgid{} and $\cgid{}^\prime$ are pairs of different nodes in \vrts that share the common label \lblop, or pairs of identical nodes \cg if the extension corresponding to \lblop is only a solution for \cg, \viz, $\edges =$ $\{$
	$(\cgr{\rulid{2}},\cgr{\rulid{4}},\{\exC\}),$
	$(\cgr{\rulid{8}},\cgr{\rulid{8}},\{\nottt \exK\})$,
	$(\cgr{\rulid{14}},\cgr{\rulid{14}},\{\nottt\neg\exT,\nottt\exT\})$
	$\} \cup F \cup H$ where
	$F = \{(\cgr{\rul},\cgr{\rul^\prime},\{\nottt \exF\})\mid \rul,\rul^\prime\in\{\rulid{4}, \rulid{6}, \rulid{10}, \rulid{11}, \rulid{13}\}, \rul\neq\rul^\prime\}$ and
	$H = \{(\cgr{\rul},\cgr{\rul^\prime},\{\nottt \exH\})\mid \rul,\rul^\prime\in\{\rulid{6}, \rulid{13}, \rulid{14}\},\rul\neq\rul^\prime\}$.

	The graphical representation of the resulting \lambdasol-graph \gcgcovex is displayed in Figure~\ref{fig:mgraph}.
\end{example}

\makeatletter
\begin{figure}
	\centering
	\begin{tikzpicture}[
			auto=center,
			group/.style={circle,fill=black!7,draw},
			label/.style={fill=black!12},
			weight/.style={above,left,color=red,yshift=4, xshift=0},
			selfloop/.style={
					to path={
							\pgfextra{
								\let\tikztotarget=\tikztostart
							}[looseness=5]\tikz@to@curve@path},font=\sffamily\small
				}]
		\node[group] (c2) at (5,0)  {$\cgr{\rulid{2}}$};
		\node[group] (c4) at (0,0)  {$\cgr{\rulid{4}}$};
		\node[group] (c6) at (6.5,3){$\cgr{\rulid{6}}$};
		\node[group] (c8) at (6.5,0)  {$\cgr{\rulid{8}}$};
		\node[group] (c10) at (2.5,4.5){$\cgr{\rulid{10}}$};
		\node[group] (c11) at (2.5,1.5) {$\cgr{\rulid{11}}$};
		\node[group] (c13) at (0,6) {$\cgr{\rulid{13}}$};
		\node[group] (c14) at (6.5,6) {$\cgr{\rulid{14}}$};
		\path[-,draw,thick]
		(c2) edge[color=purple,bend left=5] node[label] {$\exC$} (c4)
		(c4) edge[color=orange,bend right=15] node[label] {$\nottt \exF$} (c6)
		(c4) edge[color=orange,bend left=15] node[label] {$\nottt \exF$} (c10)
		(c4) edge[color=orange] node[label] {$\nottt \exF$} (c11)
		(c4) edge[color=orange,bend left=5] node[label] {$\nottt \exF$} (c13)
		(c6) edge[color=orange, bend left=7.5] node[label] {$\nottt \exF$} (c10)
		(c6) edge[color=orange, bend right=7.5] node[label] {$\nottt \exF$} (c11)
		(c6) edge[color=teal,bend right=37.5] node[label] {$\nottt \exH$} (c13)
		(c6) edge[color=orange,bend right=15] node[label] {$\nottt \exF$} (c13)
		(c6) edge[color=teal, bend right=5] node[label] {$\nottt \exH$} (c14)
		(c8) edge[selfloop,color=olive] node[label] {$\nottt \exK$} (c8)
		(c10) edge[color=orange,bend left=5] node[label] {$\nottt \exF$} (c11)
		(c10) edge[color=orange] node[label] {$\nottt \exF$} (c13)
		(c11) edge[color=orange,bend left=15] node[label] {$\nottt \exF$} (c13)
		(c13) edge[color=teal,bend left=15] node[label] {$\nottt \exH$} (c14)
		(c14) edge[selfloop,color=cyan] node[label] {$\nottt \neg \exT, \nottt \exT$} (c14);
		\draw
		(c2.west) node[weight] {1}
		(c4.west) node[weight] {1}
		(c6.east) node[xshift=4mm,weight] {1}
		(c8.east) node[xshift=4mm,weight] {1}
		(c10.west) node[weight] {1}
		(c11.west) node[weight] {1}
		(c13.west) node[weight] {1}
		(c14.east) node[xshift=4mm,weight] {2};
	\end{tikzpicture}
	\caption{\lambdasol-graph \gcgcovex}\label{fig:mgraph}
\end{figure}
\makeatother

The graph illustrates several complete subgraphs which are sets of nodes where all nodes are connected to each other by an edge with the same label.
Such subgraphs we call \emph{\lambdasol-cliques}.

\begin{definition}[\lambdasol-clique]\label{def:clique}
	Suppose a logic program \pp and a conflict group cover $\cgcov \in \ccgpp$.
	A \emph{\lambdasol-clique} \wrt a label \lblop in a \lambdasol-graph $\gcgcov = (\vrts,\edges)$ is a maximal subgraph $\gcgclq = (\vrts^\prime,\edges^\prime)$ of \gcgcov where (1) $\edges^\prime \subseteq \edges$ contains all edges in \edges with label \lblop, and (2) $\vrts^\prime \subseteq \vrts$ contains every node that is connected to an edge in $\edges^\prime$.
	We define the \emph{weight of a \lambdasol-clique} as the sum of the weights of all nodes that occur in the \lambdasol-clique.
	The set of all \lambdasol-cliques in a \lambdasol-graph \gcgcov is denoted by \clqcg.
\end{definition}

In the following, given a set of edges \edges, its subset of all edges with label \lblop is denoted by \edgeslbl{\lblop}.

Since a \lambdasol-clique \gcgclq contains all edges of \edgeslbl{\lblop} and all their connected nodes, every \lambdasol-clique is a complete graph.

\begin{example}[Example~\ref{ex:run:induce} \contd]\label{ex:run:cliques}
	$\gcgcovex = (\vrts,\edges)$ contains the following five \lambdasol-cliques:
	\begin{itemize}[leftmargin=4ex]
		\item[\cqid{1}] = $(\{(\cgr{\rulid{6}},1),(\cgr{\rulid{13}},1), (\cgr{\rulid{14}},2)\},\edgeslbl{\nottt\exH},\{\nottt \exH\})$ with weight 4
		\item[\cqid{2}] = $(\{(\cgr{\rulid{4}},1),(\cgr{\rulid{6}},1),(\cgr{\rulid{10}},1),(\cgr{\rulid{11}},1),(\cgr{\rulid{13}},1)\},\edgeslbl{\nottt\exF},\{\nottt \exF\})$ with weight 5
		\item[\cqid{3}] = $(\{(\cgr{\rulid{2}},1), (\cgr{\rulid{4}},1)\},\edgeslbl{\nottt\exC},\{\exC\})$ with weight 2
		\item[\cqid{4}] = $(\{(\cgr{\rulid{8}},1)\},\edgeslbl{\nottt\exK},\{\nottt \exK\})$ with weight 1
		\item[\cqid{5}] = $(\{(\cgr{\rulid{14}},2)\},\edgeslbl{\nottt \exT,\nottt \neg \exT},\{\nottt \exT,\nottt \neg \exT\})$ with weight 2
	\end{itemize}
	Table~\ref{tab:tab-p} shows which conflicts are involved in the different cliques
\end{example}

As all conflicts represented in a \lambdasol-clique $(\vrts^\prime,\edges^\prime)$ share a \lambdasol-extension \lambdasol, its weight $\cqweight$ indicates that $\cqweight$ different conflicts can be solved by extending \body of each representative rule $\rul$ in $(\cgr{\rul},\cgsize{})\in\vrts^\prime$ by \lblop.

To find sets of cliques such that every conflict in a program is considered, we introduce \emph{clique covers} for \lambdasol-graphs.
\begin{definition}[Clique cover]\label{def:cliquecover}
	Suppose a logic program \pp and a set $\cqb\subseteq\clqcg$ of \lambdasol-cliques in a graph $\gcgcov = (\vrts,\edges)$ are given.
	We say \cqb is a \emph{clique cover} for \gcgcov if every node in \vrts appears in at least one clique of \cqb.
	Moreover, \cqb is the \emph{minimal} clique cover if there is no clique cover $\cqb^\prime$ for \gcgcov \sT $|\cqb^\prime|\lneq|\cqb|$.
\end{definition}

A minimal clique cover for a graph \gcgcov, therefore, provides us with minimal compositions of cliques where every conflict is considered.
For this reason, our approach uses clique covers to determine which conflicts can be solved by the same \lambdasol-extensions which in turn can be used to find a suitable order in which conflicts and their solutions can be suggested to the user.
Herewith, we arrive at the following result which will be useful in the following.
\begin{proposition}\label{prop:mincliquecover}
	Given a program \pp with conflicts and a \lambdasol-clique $\cgcov\in\ccgpp$, a minimal clique cover for \gcgcov provides a minimal set of \lambdasol-extensions $\mathbf{L} = \{\lblop \mid (\vrts^\prime,\lblop)\in\cqb\}$ that is required to obtain a program without conflicts.
\end{proposition}

\begin{proof}
	Let $\cgcov\in\ccgpp$ be a \lambdasol-clique of a program $\pp=(\vrts,\edges)$ with conflicts, and $\cqb\subseteq\clqcg$ a clique cover for \gcgcov.
	By Definition~\ref{def:graph}, $V$ contains a node for each conflict group of a conflict group cover \cgcov.
	By Proposition~\ref{prop:groupcover}, \cgcov considers all conflicting rules of \pp.
	Thus, every conflict in \pp is implicitly represented by at least one node in \vrts.
	By Definition~\ref{def:clique}, a \lambdasol-clique $(\vrts^\prime,\edgeslbl{\lblop},\lblop)$ in \gcgcov condenses conflict groups that share a common \lambdasol-extension \lblop, and by Definition~\ref{def:cliquecover}, each node in $V$ appears in at least one clique of \cqb.
	Hence, \cqb implies a set of \lambdasol-extensions $\mathbf{L} = \{\lblop \mid (\vrts^\prime,\lblop)\in\cqb\}$ that are sufficient to resolve all conflicts in \pp, by extending the body of representative rules in the conflict groups suitably.
	Likewise if such a clique cover is minimal, \cgcov implies the smallest set of solutions for \pp.
\end{proof}

The notion of minimal cover is illustrated in the following example.

\begin{example}[Example~\ref{ex:run:induce} \contd]\label{ex:run:graphcover}
	The minimal clique cover in \gcgcovex is the set $\{\cqid{1},\cqid{2},\cqid{3},\cqid{4}\}$.
	Therefore, the four \lambdasol-extensions $\{\exC\}$, $\{\nottt\exF\}$, $\{\nottt \exH\}$, and $\{\nottt \exK\}$ suffice to obtain a program without conflicts.
\end{example}

Now we are ready to show how the results provide the crucial basis to define suitable orders over conflicts and \lambdasol-extensions for their usage in conflict resolution frameworks.

%===========================================================
\section{Sorting conflicts and \lambdasol-extensions}%
\label{sec:sort-conflict-using-lambda-graphs}
%===========================================================
In this section we show how \lambdasol-graphs and clique covers can be utilized to define strategies that compute (1) an order over all conflict groups of a program, and (2) an order over the \lambdasol-solutions of each conflict group.
These orders can then be used to define in which sequence they should be presented to the knowledge expert.
The goal is to improve the efficiency of the interactive conflict resolution process with the expert and to facilitate the resolution process overall by prioritizing those conflict groups whose solutions can be used for other conflict groups and thereby solve the most amount of conflicts simultaneously.
For that reason in this section, we provide an example of such a strategy that makes use of the technical notions introduced in the preceding sections.

We begin by introducing the notion of \emph{relationships}:
We say a \lambdasol-clique \emph{$\cq$ is related to another \lambdasol-clique $\cq^\prime$} iff $\cq$ and $\cq^\prime$ share a common node.
Moreover, we say that a conflict group $\cg$ is \emph{part of a \lambdasol-clique $\cq$} if the node that corresponds to \cg is in \cq.
With these conventions, we are now able to define a strategy to sort conflicts and their \lambdasol-extensions such that the user can resolve all conflicts in a more suitable way.

%===========================================================
\subsection{Order strategy}%
\label{subsec:order-strategy}
%===========================================================
\begin{algorithm}[tb]
	\caption{Interactive conflict resolution}
	\label{alg:algorithm}
	\textbf{Input}: Logic program \pp with conflicts\\
	\textbf{Output}: Uniformly non-contraditory logic program \ppp
	\begin{algorithmic}[1] %[1] enables line numbers
		\STATE Let $\ppp = \pp$.
		\WHILE{$\ppp$ has conflicts}\label{line:whileincons}
		\STATE Choose a suitable conflict group cover $\cgcov\in\ccgppp$.\label{line:choosecov}
		\STATE Generate corresponding \lambdasol-graph \gcgcov.
		\STATE Choose a suitable clique cover $\cqb\subseteq\clqcg$.\label{line:chooseclq}
		\STATE Compute \cgord~: Sort conflict groups by the number of related \lambdasol-cliques in ascending order.\label{line:step1a}
		\FOR{Element $e$ in \cgord}
		\IF{$e$ contains more than one conflict group}
		\STATE Extend order \cgord~: sort conflict groups in $e$ by their weight in descending order.
		\STATE Sort conflict groups with the same weight in alphanumerical order.
		\ENDIF
		\ENDFOR
		\FOR{Element \cgr{\rul} in \cgord}
		\STATE Compute $\lambdaord{\cgr{\rul}}$: Sort \lambdasol-extensions of the representative of \cgr{\rul} by their \lambdasol-clique's weight.
		\FOR{Element $e$ in $\lambdaord{\cgr{\rul}}$}
		\IF{$e$ contains multiple \lambdasol-extensions}
		\STATE Extend \lambdaord{\cgr{\rul}}: Sort \lambdasol-extensions in $e$ in alphanumerical order.
		\ENDIF
		\ENDFOR
		\ENDFOR
		\STATE Present expert the conflict groups and the \lambdasol-extensions of the representatives in their respective order.
		\IF{Expert chooses a \lambdasol-extension for a conflict group \cgr{\rul}}
		\STATE In \ppp, replace \rul by the corresponding \lambdasol-extended rule of \rul.
		\ENDIF
		\ENDWHILE
		\STATE \textbf{return} \ppp
	\end{algorithmic}
\end{algorithm}

The goal of an order strategy is to establish an order in which the conflict groups are presented to the user and to additionally obtain an order for each conflict group that specifies how the respective \lambdasol-extensions are suggested.
The primary objective is to assist the knowledge expert to efficiently find the correct resolution for each conflict by preferring solutions that solve the most amount of conflicts simultaneously.

Recall that a clique cover of \gcgcov contains a set of conflict groups that considers all conflicts in \pp and, by definition, also implicitly provides possible solutions for every conflict.
The clique cover furthermore specifies explicitly which rules of \pp will be modified.
The following strategy will illustrate how these properties can be used to obtain an order over conflict groups and \lambdasol-extensions.

The first step employs two sorting criteria in lexicographical order.
First, a general order over the conflict groups in \cgcov is determined by viewing the number of edges of each node in \gcgcov in order to prefer conflict groups with less possible solutions.
Conflicts with only one possible solution can hereby be dealt with first which can potentially reduce the complexity of subsequent conflict resolutions.
However, since many conflict groups can have the same amount of solutions, we refine this order in a subsequent action  by taking the corresponding clique weights into account.

The second step of the strategy determines an order over \lambdasol-extensions for each conflict group.
Here again, clique weights are utilized.
Consequently, this strategy defines an order over conflict groups and possible solutions where those groups with the least possible solutions are presented first and for each conflict group those solutions are preferred that can be used to resolve related conflicts in parallel.

We now define these steps in more detail.
Each step is explained by means of our running example.

\paragraph{Step 1a:}
In the first step, the conflict groups are ordered by the number of cliques they are part of.
The user is shown those conflict groups first that are related to the least amount of cliques.
By this, the knowledge expert is being presented with as few choices at a time as possible.

\begin{example}[Example~\ref{ex:run:graphcover} \contd]\label{ex:run:step-1}
	The respective representatives of conflict groups \cgr{\rulid{2}}, \cgr{\rulid{8}}, \cgr{\rulid{10}}, and \cgr{\rulid{11}} each have only one possible \lambdasol-extension.
	The representatives of all remaining conflict groups have two.
	For \gcgcovex, we thus get the preliminary order of conflict groups
	$$\cgr{\rulid{2}},\cgr{\rulid{8}},\cgr{\rulid{10}},\cgr{\rulid{11}} \cgord \cgr{\rulid{4}},\cgr{\rulid{6}},\cgr{\rulid{13}},\cgr{\rulid{14}}$$
	that says that all conflict groups on the left side should be presented before those on the right.
\end{example}

It is easy to see that even in smaller programs, this kind of order can be too coarse-grained.
To order conflict groups that have the same amount of possible solutions, we propose the following subsequent step.

\paragraph{Step 1b:}
To refine the order obtained in Step 1a, we use the weight of the cliques.
As mentioned before, the weight of a clique represents the amount of different conflict groups that can be solved by the \lambdasol-extension that is represented by the label of the clique.
Therefore, the conflict groups with the same amount of solutions should additionally be arranged by the sum of weights of all cliques they are part of in descending order.
Thereby, we provide the knowledge expert with the opportunity to resolve as many conflicts as possible as soon as possible.
If there are conflict groups with the same total weight, we simply arrange them in alphanumerical order.

\begin{example}[Ex.~\ref{ex:run:step-1} \contd]\label{ex:run:step-2}
	Conflict group \cgr{\rulid{2}} is only part of clique \cqid{3} that has weight 2.
	Conflict group \cgr{\rulid{8}} is only part of clique \cqid{4} with weight 1.
	Conflict groups \cgr{\rulid{10}} and \cgr{\rulid{11}} are only part of clique \cqid{2} that has weight  5.
	For conflict group \cgr{\rulid{4}}, which is both in \cqid{2} and \cqid{3}, we get the total weight of 7 as \cqid{3} has weight  2 and \cqid{2} has weight  5.
	Likewise for \cgr{\rulid{6}} and \cgr{\rulid{13}}, we get a total weight of 9, and for \cgr{\rulid{14}} a total weight of 6.
	This way, for \gcgcovex we obtain the following, more specific order:
	\begin{align*}
		\cgr{\rulid{10}} \cgord \cgr{\rulid{11}} \cgord \cgr{\rulid{2}} \cgord \cgr{\rulid{8}} \cgord \cgr{\rulid{6}} \cgord \cgr{\rulid{13}} \cgord \cgr{\rulid{4}} \cgord \cgr{\rulid{14}}
	\end{align*}
\end{example}

\paragraph{Step 2:}
Step 1 provides us with a suitable order over conflict groups.
As conflict groups can have multiple \lambdasol-extensions (see Example~\ref{ex:run:multiple}), Step 2 defines how one can obtain an order $\lambdaord{\cg}$ over all \lambdasol-extensions for the representative of each conflict group \cg.
For that we will use the weight of their respective cliques.
That is, the \lambdasol-extensions for the representative in each group are ordered by the weight of their respective clique in descending order.
If for two extensions the clique weight is identical, again, we sort them in alphanumerical order.

\begin{example}[Ex.~\ref{ex:run:step-2} \contd]\label{ex:run:step-3}
	The representative of conflict groups \cgr{\rulid{2}}, \cgr{\rulid{8}}, \cgr{\rulid{10}}, and \cgr{\rulid{11}} each have their own unique solution, \viz, $\{\exC\}$, $\{\nottt \exK\}$, and $\{\nottt \exF\}$.
	According to Step 2 for \cgr{\rulid{6}}, $\{\nottt \exF\}$ should be suggested first as its clique has weight 5, and, if the expert does not accept the first extension, $\{\nottt \exH\}$ whose clique has size 4 can be presented as an alternative.
	The same order also holds for \cgr{\rulid{13}} as the possible \lambdasol-extensions of their representatives are identical.
	For \cgr{\rulid{4}}, $\{\nottt \exF\}$ has also to be suggested first, then the remaining solution $\{\exC\}$.
	Similarly, for \cgr{\rulid{14}}, $\{\nottt \exH\}$ has to be presented first and $\{\nottt \exT, \nottt \neg \exT\}$ after that.
	As a result, for those conflict groups with multiple extensions, we gain the following orders:
	\begin{align*}
		 & \{\nottt\exF\}\lambdaord{\cgr{\rulid{4}}}\{\exC\}, &  & \{\nottt\exF\}\lambdaord{\cgr{\rulid{6}}}\{\nottt\exH\}, &  & \{\nottt\exF\}\lambdaord{\cgr{\rulid{13}}}\{\nottt\exH\}, &  & \{\nottt\exH\}\lambdaord{\cgr{\rulid{14}}}\{\nottt\exT,\nottt\neg\exT\}
	\end{align*}
\end{example}

Keep in mind that since the expert can apply a \lambdasol-extension on conflicts of multiple conflict groups simultaneously, some of the subsequent conflict groups and solutions can become obsolete.
It is therefore necessary to recalculate the sorting of the remaining conflict groups and solutions once the expert accepted a \lambdasol-extension to be applied as the input program changes, in other words, after each program modification the remaining conflict groups have to be identified and the order over conflict groups and \lambdasol-extensions have to be computed accordingly.
The computation of \lambdasol-graphs and all its orders after each program modification is however expedient as it provides the means to implement the crucial functionality to postpone the resolution of certain conflicts as the orders of conflict groups and \lambdasol-extension are only a recommendation based on the syntactical properties of the program.

Algorithm~\ref{alg:algorithm} summarizes the complete workflow.
We recommend that in an actual implementation, the choices stated in lines~\ref{line:choosecov} and~\ref{line:chooseclq} should be made in interaction with the expert by presenting the conflicts in an appropriate fashion as Example~\ref{ex:run:extensions} shows that choosing the representatives of conflicts is a crucial step and not straightforward especially if multiple rules of a conflict are eligible for modification by a \lambdasol-extension.
An explicit implementation should therefore provide the ability to revise the chosen conflict groups and representatives if the final program is not deemed satisfactory.

To conclude this section, we illustrate the workings of this strategy by applying it on the running example.
This example will also propose possible ways to provide the knowledge expert with additional information using the \lambdasol-graph and simulate a possible line of thought of the expert.

\begin{example}[Example~\ref{ex:run:step-3} \contd]\label{ex:run:strat}
	Assume that a knowledge expert is assigned to resolve all conflicts in \ppEx.
	According to the conflict group order obtained in Example~\ref{ex:run:step-2}, conflict group \cgr{\rulid{10}} is presented first.
	As stated in Table~\ref{tab:tab-p}, \rulid{10} only has the possible solution $\{\nottt \exF\}$.
	The \lambdasol-graph of \gcgcovex shows that this solution belongs to \lambdasol-clique \cqid{5} (see Example~\ref{ex:run:cliques}) which states that the representatives of conflict groups \cgr{\rulid{4}}, \cgr{\rulid{6}}, \cgr{\rulid{11}}, and \cgr{\rulid{13}} all have $\{\nottt\exF\}$ as a possible solution.
	By showing these connections, the knowledge expert can decide whether they want to apply the extension $\{\nottt\exF\}$ to to the representatives of other conflict groups of \cqid{5} as well.
	Suppose that the expert knows of a connection between the property encoded in \exF and those encoded in \exY, \exW, and \exP.
	Regarding atom \exX and rules \rulid{3} and \rulid{4}, the expert sees no immediate connection.%, therefore passing over \cgr{\rulid{4}}.
	They therefore decide to apply the \lambdasol-extension $\nottt\exF$ not only to \rulid{10}, but also to \rulid{6}, \rulid{11}, and \rulid{13}.
	This action leaves the set $\cgcovprime = \{\cgr{\rulid{2}},\cgr{\rulid{4}},\cgr{\rulid{8}},\cgr{\rulid{14}}\}$ of unresolved conflict groups with the resulting \lambdasol-graph $\gcgcovprime=$ $(\{(\cgr{\rulid{2}},1),$ $(\cgr{\rulid{4}},1),$ $(\cgr{\rulid{8}},1),$ $(\cgr{\rulid{14}},2)\},$ $\{(\cgr{\rulid{2}},$ $\cgr{\rulid{4}},\{\exC\}),$ $(\cgr{\rulid{8}},\cgr{\rulid{8}},\{\nottt\exK\}),$ $(\cgr{\rulid{14}},\cgr{\rulid{14}},\{\nottt\exH\}),$ $(\cgr{\rulid{14}},\cgr{\rulid{14}},\{\nottt\exT,\nottt\neg\exT\})\}$
	Note that since the conflicts of \cgr{\rulid{6}} and \cgr{\rulid{13}} are resolved now, there is an additional self-loop $(\cgr{\rulid{14}},\cgr{\rulid{14}},\{\nottt\exH\})$ for \cgr{\rulid{14}}.
	This reduced \lambdasol-graph therefore leads to the following conflict group order
	$$\cgr{\rulid{2}} \cgordprime \cgr{\rulid{4}} \cgordprime \cgr{\rulid{8}} \cgordprime \cgr{\rulid{14}}.$$
	As before, along with \cgr{\rulid{2}}, the expert is also presented with \cgr{\rulid{4}} as they both still build up \lambdasol-clique \cqid{3} and their representatives share the common \lambdasol-extension $\{\exC\}$.
	Let the knowledge expert apply $\{\exC\}$ to both representatives, thereby reflecting that $\{\exC\}$ is a more suitable solution for resolving the conflict in \cgr{\rulid{4}} than $\{\nottt\exF\}$.
	This leaves the expert with the last conflict group cover $\cgcovprimeprime = \{\cgr{\rulid{8}},\cgr{\rulid{14}}\}$ and the resulting \lambdasol-graph $\gcgcovprimeprime=$ $(\{(\cgr{\rulid{8}},1),$ $(\cgr{\rulid{14}},2)\},$ $\{(\cgr{\rulid{8}},\cgr{\rulid{8}},\{\nottt\exK\}),$ $(\cgr{\rulid{14}},\cgr{\rulid{14}},\{\nottt\exH\}),$ $(\cgr{\rulid{14}},\cgr{\rulid{14}},\{\nottt\exT,\nottt\neg\exT\})\}$.
	For \cgr{\rulid{8}}, the expert chooses the only solution $\{\nottt k\}$.
	For the conflicts of \cgr{\rulid{14}}, the expert can lastly choose between $\{\nottt\exH\}$ and $\{\nottt\exT,\nottt\neg\exT\}$ as \lambdasol-extensions.
	Originally, the order over these extensions was $\{\nottt\exH\}\lambdaord{\cgr{\rulid{14}}}\{\nottt\exT,\nottt\neg\exT\}$ (see Example~\ref{ex:run:step-3}) because \lambdasol-clique \cqid{1} has a higher weight than \cqid{5}.
	Now that all conflicts are resolved except for those of \cgr{\rulid{14}} and the weight of the \lambdasol-clique regarding $\nottt\exH$ decreased from 4 to 2, the order between the two \lambdasol-extensions is determined by their alphanumerical order.
	In this case, the alphanumerical order coincides with the original order.
	Thus, for \cgr{\rulid{14}}, the expert is first presented with \lambdasol-extension $\{\nottt\exH\}$ which they immediately choose as the most fitting solution.
	This concludes the conflict resolution process that outputs the following conflict-free program:

	\begin{align*}
		 & \rulrow{1} \exA \when \exB,\nottt \exC.                   &
		 & \rulrow{2} \neg \exA \when \exB,\exC.                     &
		 & \rulrow{3} \exX \when \exD,\exE,\exF,\nottt \exC.         &
		 & \rulrow{4} \neg \exX \when \exD,\exE,\exC.                & \\
		 & \rulrow{5} \exY \when \exG,\exH,\exF.                     &
		 & \rulrow{6} \neg \exY \when \exG,\nottt~\exF.              &
		 & \rulrow{7} \exZ \when \exJ,\exK,\nottt \exL.              &
		 & \rulrow{8} \neg \exZ \when \exJ,\nottt \exL,\nottt\exK.   & \\
		 & \rulrow{9} \exW \when \exF,\exM,\exN.                     &
		 & \rulrow{10} \neg \exW \when \exM,\nottt~\exF.             &
		 & \rulrow{11} \neg \exW \when \exN,\nottt~\exF.             &
		 & \rulrow{12} \exP \when \exO,\exH,\exF,\nottt \exQ.        & \\
		 & \rulrow{13} \neg \exP \when \exO,\nottt \exQ,\nottt~\exF. &
		 & \rulrow{14} \exU \when \exS,\nottt\exH.                   &
		 & \rulrow{15} \neg \exU \when \exS,\neg \exT,\exH.          &
		 & \rulrow{16} \neg \exU \when \exS,\exT,\exH.               &
	\end{align*}
\end{example}

Note that Example~\ref{ex:run:strat} illustrates just one of many possible ways how ordering strategies and \lambdasol-graphs can be utilized for the implementation of interactive resolution of conflicts.
For instance, instead of completely omitting conflicts during the process once they are resolved, these conflicts can be shown further on, only flagging them as resolved.
This would provide the expert with additional information that is otherwise removed once the cliques of resolved conflict groups are removed.
Such a functionality could also be extended by the possibility to revert previous modifications.

%===========================================================
\section{Related work}%
\label{sec:related-work}
%===========================================================
The method of conflict resolution is closely related to the topic of \emph{ASP debugging}.
There we find several approaches that deal with the modification of logic programs~\cite{GebserPuehrerSchaubTompits2008,OetschPuehrerTompits2011,Shchekotykhin2015}.
These programs are not necessarily inconsistent.
They rather help the user knowledge expert to fix the mismatch between the current program's semantics and the semantics intended by the program's modeller.
In~\cite{DodaroGasteigerRealeRiccaSchekotihin2019}, the authors utilize the notion of incoherence to implement the debugging of programs.
All these approaches, however, require the knowledge expert to provide further information in order to detect and resolve the faulty parts of the program.

In~\cite{MenciaMarquesSilva2020}, the authors present an approach to resolve inconsistency (by contradictions or incoherence) by finding minimal sets of rules that are causing inconsistency.
Similar to~\cite{ThevapalanKernIsberner2022}, it also presents a way to compute possible solutions, which in this case are minimal correction sets of rules whose removal guarantee that the reduced program is consistent.
These minimal sets are identical to those found in the presented approach if the inconsistency is caused by contradictory literals.
%This paper describes ways to utilize such correction sets to gain a consistent program not by directly removing those rules but instead providing the knowledge expert the necessary information to modify the program efficiently such that the modified program is consistent and represents professionally adequate knowledge.
Compared to that work, our approach in this paper helps to preserve information by not removing those rules, but instead exploiting dependencies and subtle differences in conflicting rules so that the knowledge expert is provided with suitable information to sharpen the rules by extending them.
In this way, (potential) conflicts are resolved and actually help to make the knowledge expressed by the program more professionally adequate.

%To the best of our knowledge, there is no approach in the literature that detects conflicts on a solely syntactical level, thus not requiring any additional information from the knowledge expert, generates possible solutions and, as we propose 
In this work, we provide an interactive solution strategy by suggesting an order over the problem causes and an order over the possible solutions.
Once suitable solutions are available, both, the approach using \lambdasol-extensions as well as debugging approaches like those based on the \emph{meta-programming technique}~\cite{GebserPuehrerSchaubTompits2008} can be used to obtain a consistent program.

%===========================================================
\section{Conclusion and future work}%
\label{sec:conclusion-and-future-work}
%===========================================================
In practice, it is often imperative that a program (core) is usable with different instance data.
For example, in the medical sector, decision support systems are often required to be usable with different patient data.
Resolving conflicts in this manner ensures that program cores can be used in such real-world applications in a safely manner \wrt consistency.
This paper extends the work in~\cite{ThevapalanKernIsberner2020,ThevapalanKernIsberner2022} that propose a framework for obtaining uniformly non-contraditory logic programs by enabling knowledge experts to interactively resolve conflicts.
Since conflicting rules can have a large amount of possible solutions, methods are necessary that provide a proper order in which the conflicts of a program and their possible solutions should be presented to the knowledge expert.
This paper tackled this issue by providing the theoretical groundwork to define such sorting methods.
We furthermore provide an explicit strategy for sorting conflicts and \lambdasol-extensions to illustrate the results.
We have shown that the syntactical structure of answer set programs suffices to acquire preferences over conflicts and their solutions as it provides the necessary information to find relationships between conflicting rules and their \lambdasol-extensions.

Further investigations are needed to combine the syntax-based view to sort elements during the conflict resolution process with methods based on the (intended) semantics of the program.
For this, proper interaction with the user is needed to obtain the relevant information.
In future work we want to extend the approach of \lambdasol-extensions to consider other types of conflict groups as they are currently limited to finding solution for a single rule that is in conflict with one or multiple other rules and, for example, not the other way around (see rules \rulid{9}-\rulid{11} in Example~\ref{ex:run}).
Such results offer data to find more connections between rules and solutions and ergo for even more precise sorting strategies.
To illustrate the practical relevance and feasibility, we are currently working on an implementation of a conflict resolution framework with the presented capabilities.

\bibliographystyle{eptcs}
\bibliography{references}

\begin{thebibliography}{1}
\providecommand{\bibitemdeclare}[2]{}
\providecommand{\surnamestart}{}
\providecommand{\surnameend}{}
\providecommand{\urlprefix}{Available at }
\providecommand{\url}[1]{\texttt{#1}}
\providecommand{\href}[2]{\texttt{#2}}
\providecommand{\urlalt}[2]{\href{#1}{#2}}
\providecommand{\doi}[1]{doi:\urlalt{https://doi.org/#1}{#1}}
\providecommand{\eprint}[1]{arXiv:\urlalt{https://arxiv.org/abs/#1}{#1}}
\providecommand{\bibinfo}[2]{#2}

\bibitemdeclare{inproceedings}{CostantiniIntrigilaProvetti2003}
\bibitem{CostantiniIntrigilaProvetti2003}
\bibinfo{author}{Stefania \surnamestart Costantini\surnameend},
  \bibinfo{author}{Benedetto \surnamestart Intrigila\surnameend} \&
  \bibinfo{author}{Alessandro \surnamestart Provetti\surnameend}
  (\bibinfo{year}{2003}): \emph{\bibinfo{title}{Coherence of updates in answer
  set programming}}.
\newblock In: {\slshape \bibinfo{booktitle}{In Proc. of the IJCAI-2003 Workshop
  on Nonmonotonic Reasoning, Action and Change}}, pp. \bibinfo{pages}{66--72}.

\bibitemdeclare{article}{DodaroGasteigerRealeRiccaSchekotihin2019}
\bibitem{DodaroGasteigerRealeRiccaSchekotihin2019}
\bibinfo{author}{Carmine \surnamestart Dodaro\surnameend},
  \bibinfo{author}{Philip \surnamestart Gasteiger\surnameend},
  \bibinfo{author}{Kristian \surnamestart Reale\surnameend},
  \bibinfo{author}{Francesco \surnamestart Ricca\surnameend} \&
  \bibinfo{author}{Konstantin \surnamestart Schekotihin\surnameend}
  (\bibinfo{year}{2019}): \emph{\bibinfo{title}{Debugging Non-ground {ASP}
  Programs: Technique and Graphical Tools}}.
\newblock {\slshape \bibinfo{journal}{Theory Pract. Log. Program.}}
  \bibinfo{volume}{19}(\bibinfo{number}{2}), pp. \bibinfo{pages}{290--316},
  \doi{10.1017/S1471068418000492}.

\bibitemdeclare{inproceedings}{GebserPuehrerSchaubTompits2008}
\bibitem{GebserPuehrerSchaubTompits2008}
\bibinfo{author}{Martin \surnamestart Gebser\surnameend},
  \bibinfo{author}{J{\"{o}}rg \surnamestart P{\"{u}}hrer\surnameend},
  \bibinfo{author}{Torsten \surnamestart Schaub\surnameend} \&
  \bibinfo{author}{Hans \surnamestart Tompits\surnameend}
  (\bibinfo{year}{2008}): \emph{\bibinfo{title}{A Meta-Programming Technique
  for Debugging Answer-Set Programs}}.
\newblock In \bibinfo{editor}{Dieter \surnamestart Fox\surnameend} \&
  \bibinfo{editor}{Carla~P. \surnamestart Gomes\surnameend}, editors: {\slshape
  \bibinfo{booktitle}{Proceedings of the Twenty-Third {AAAI} Conference on
  Artificial Intelligence, {AAAI} 2008, Chicago, Illinois, USA, July 13-17,
  2008}}, \bibinfo{publisher}{{AAAI} Press}, pp. \bibinfo{pages}{448--453}.

\bibitemdeclare{article}{GelfondLifschitz1991}
\bibitem{GelfondLifschitz1991}
\bibinfo{author}{Michael \surnamestart Gelfond\surnameend} \&
  \bibinfo{author}{Vladimir \surnamestart Lifschitz\surnameend}
  (\bibinfo{year}{1991}): \emph{\bibinfo{title}{Classical Negation in Logic
  Programs and Disjunctive Databases}}.
\newblock {\slshape \bibinfo{journal}{New Gener. Comput.}}
  \bibinfo{volume}{9}(\bibinfo{number}{3/4}), pp. \bibinfo{pages}{365--386},
  \doi{10.1007/BF03037169}.

\bibitemdeclare{inproceedings}{MenciaMarquesSilva2020}
\bibitem{MenciaMarquesSilva2020}
\bibinfo{author}{Carlos \surnamestart Menc{\'{\i}}a\surnameend} \&
  \bibinfo{author}{Jo{\~{a}}o \surnamestart Marques{-}Silva\surnameend}
  (\bibinfo{year}{2020}): \emph{\bibinfo{title}{Reasoning About Strong
  Inconsistency in {ASP}}}.
\newblock In \bibinfo{editor}{Luca \surnamestart Pulina\surnameend} \&
  \bibinfo{editor}{Martina \surnamestart Seidl\surnameend}, editors: {\slshape
  \bibinfo{booktitle}{Theory and Applications of Satisfiability Testing - {SAT}
  2020 - 23rd International Conference, Alghero, Italy, July 3-10, 2020,
  Proceedings}}, {\slshape \bibinfo{series}{Lecture Notes in Computer Science}}
  \bibinfo{volume}{12178}, \bibinfo{publisher}{Springer}, pp.
  \bibinfo{pages}{332--342}, \doi{10.1007/978-3-030-51825-7\_24}.

\bibitemdeclare{inproceedings}{OetschPuehrerTompits2011}
\bibitem{OetschPuehrerTompits2011}
\bibinfo{author}{Johannes \surnamestart Oetsch\surnameend},
  \bibinfo{author}{J{\"{o}}rg \surnamestart P{\"{u}}hrer\surnameend} \&
  \bibinfo{author}{Hans \surnamestart Tompits\surnameend}
  (\bibinfo{year}{2011}): \emph{\bibinfo{title}{Stepping through an Answer-Set
  Program}}.
\newblock In \bibinfo{editor}{James~P. \surnamestart Delgrande\surnameend} \&
  \bibinfo{editor}{Wolfgang \surnamestart Faber\surnameend}, editors: {\slshape
  \bibinfo{booktitle}{Logic Programming and Nonmonotonic Reasoning - 11th
  International Conference, {LPNMR} 2011, Vancouver, Canada, May 16-19, 2011.
  Proceedings}}, {\slshape \bibinfo{series}{Lecture Notes in Computer Science}}
  \bibinfo{volume}{6645}, \bibinfo{publisher}{Springer}, pp.
  \bibinfo{pages}{134--147}, \doi{10.1007/978-3-642-20895-9\_13}.

\bibitemdeclare{inproceedings}{Shchekotykhin2015}
\bibitem{Shchekotykhin2015}
\bibinfo{author}{Kostyantyn~M. \surnamestart Shchekotykhin\surnameend}
  (\bibinfo{year}{2015}): \emph{\bibinfo{title}{Interactive Query-Based
  Debugging of {ASP} Programs}}.
\newblock In \bibinfo{editor}{Blai \surnamestart Bonet\surnameend} \&
  \bibinfo{editor}{Sven \surnamestart Koenig\surnameend}, editors: {\slshape
  \bibinfo{booktitle}{Proceedings of the Twenty-Ninth {AAAI} Conference on
  Artificial Intelligence, January 25-30, 2015, Austin, Texas, {USA}}},
  \bibinfo{publisher}{{AAAI} Press}, pp. \bibinfo{pages}{1597--1603}.

\bibitemdeclare{inproceedings}{ThevapalanKernIsberner2020}
\bibitem{ThevapalanKernIsberner2020}
\bibinfo{author}{Andre \surnamestart Thevapalan\surnameend} \&
  \bibinfo{author}{Gabriele \surnamestart Kern-Isberner\surnameend}
  (\bibinfo{year}{2020}): \emph{\bibinfo{title}{Towards Interactive Conflict
  Resolution in {ASP} Programs}}.
\newblock In \bibinfo{editor}{Maria~Vanina \surnamestart Martínez\surnameend}
  \& \bibinfo{editor}{Ivan \surnamestart Varzinczak\surnameend}, editors:
  {\slshape \bibinfo{booktitle}{Proceedings of the 18th International Workshop
  on Non-Monotonic Reasoning, {NMR 2020}}}, pp. \bibinfo{pages}{29--36}.

\bibitemdeclare{article}{ThevapalanKernIsberner2022}
\bibitem{ThevapalanKernIsberner2022}
\bibinfo{author}{Andre \surnamestart Thevapalan\surnameend} \&
  \bibinfo{author}{Gabriele \surnamestart Kern-Isberner\surnameend}
  (\bibinfo{year}{2022}): \emph{\bibinfo{title}{On Establishing Robust
  Consistency in Answer Set Programs}}.
\newblock {\slshape \bibinfo{journal}{Theory and Practice of Logic
  Programming}}, p. \bibinfo{pages}{1–34}, \doi{10.1017/S1471068422000357}.

\end{thebibliography}

\end{document}